\documentclass[a4,11pt]{article}
\newtheorem{lemma}{Lemma}
\newtheorem{theorem}{Theorem}
\newenvironment{proof}%
{\begin{trivlist}\item[\hspace*{\labelsep}{\it Proof.\/}]}%
{\hfill$\Box{}$\end{trivlist}}

\newtheorem{claim}{Claim}

\newtheorem{fact}{Fact}

\newcommand{\eps}{\varepsilon}

\newcommand{\para}{\medskip\noindent}

\newcommand{\head}[1]
 {\markright{\hbox to 0pt{\vtop to 0pt{\hbox{}\vskip 3mm \hrule
 width  \textwidth \vss} \hss}{\sc #1}}}
\usepackage{latexsym,graphicx,amsmath}
\begin{document}
\title{\bf On Two Dimensional Orthogonal Knapsack Problem } 
\author{Xin Han$^1$  \hspace{3mm} Kazuo Iwama$^1$\hspace{3mm} Guochuan
 Zhang$^{2}$
  \\ {\small School of Informatics, Kyoto University, Kyoto
 606-8501, Japan} \\ {\small hanxin, iwama@kuis.kyoto-u.ac.jp}
 \\ {\small $^2$ Department of Mathematics, Zhejiang University, China}
 \\ {\small zgc@zju.edu.cn}
}
\date{}
\maketitle 


\begin{abstract}
 In this paper, we study the following knapsack problem:
 Given a list of squares with profits, 
 we are requested to pack a sublist of them into a rectangular bin 
  (not a  unit square bin) to make  profits in the bin as large as possible. 
  We first observe there is a  Polynomial Time Approximation Scheme (PTAS) 
  for the problem of packing weighted squares into
rectangular bins with large resources,
then apply the PTAS to the problem of packing
squares with profits into a rectangular bin and
get a   $\frac65+\epsilon$ approximation algorithm.
\end{abstract}

\section{Introduction}



 In this paper we study the two-dimensional generalization for 
 {\em Knapsack}: we are given a set of squares, each of which
 is associated with a profit. The goal is to pack a subset of
 the squares (items) into  a rectangle (bin) to maximize the total profit
  packed. The problem is NP-hard in the strong
 sense even if each  item is an unweighted square
 (i.e., its profit is equal to its area) \cite{Leung1990}.
 A little surprisingly, the research for approximation algorithm
 has started quite recently: Jansen and Zhang \cite{jz04},
  Caprara and Monaci \cite{cm04}, Harren \cite{Harren06} etc.

 {\bf Related Work}
 There are many literatures on rectangle packing and square packing.
 For a two dimensional knapsack problem
 in which a subset of a given set of rectangles
 are  packed  into  a given rectangular bin to maximize the total
 profits in the bin.
  Jansen and Zhang proposed $2+\epsilon$ approximation algorithm \cite{jz04}.
 When all items are squares and their profits are equal to their areas,
 Fishkin, Gerber, Jansen and Solis-Oba \cite{FGJS05} 
 presented a PTAS,
 which was also obtained by Han, Iwama and Zhang independently \cite{HIZ05}.
  Jansen and Zhang \cite{JZSWAT04} proposed a PTAS for 
  packing squares into a rectangular bin to maximize 
  the number of squares packed in the bin \cite{JZSWAT04}.
 Harren \cite{Harren06} 
 proposed $\frac54+\epsilon$ approximation algorithm
 for packing squares into a {\em unit square} bin.
 But his algorithm is not applicable to pack squares into 
 a  {\em rectangular} bin since 
  his algorithm requires that every side of the bin must have the same length.
 Fishkin, Gerber and Jansen \cite{FGJ04} obtained a $(1-\epsilon)$-approximation algorithm 
 for packing a set of rectangles with profits into a large resource bin 
with width 1 and height larger than $(1/\epsilon)^{4}$.

 Another related work is 2 dimensional bin packing problem in which
 all rectangles have to be packed into a unit square bin to minimize
 the number of bins required.
 When all items are squares,
Ferreira et al.~\cite{fmw99} gave an approximation algorithm with
asymptotic worst-case ratio bounded above by 1.988. Kohayakawa et
al. \cite{kmrw01} and Seiden and van Stee \cite{ss03}
independently obtained approximation algorithms with asymptotic
worst-case ratio of at most $14/9+\eps$ (for any $\eps>0$). These
results were recently improved by Bansal, Correa, Kenyon 
 and Sviridenko \cite{BS06}. They  proposed
an asymptotic PTAS for packing $d$-dimensional cubes into the
minimum number of unit cubes. For the online case,
if the number of bins is unbounded, the best known
asymptotic worst case ratio is 2.1439 \cite{Han2006}.

There are also some research on the Multiple Knapsack Problem.
Kellerer \cite{Kell99},
first gave a PTAS for a special case of this problem in which all
the knapsacks have identical capacity \cite{Kell99}. Chekuri and
Khanna \cite{CS05} obtained a PTAS for the general multiple
knapsack problem.
For packing rectangles into multiple identical rectangular bins,
Fishkin et al. \cite{FGJ04} gave a $2+\epsilon$ approximation algorithm.

 {\bf Main results and Techniques:}
We first observe that 
the techniques used in Mutilple Knapsack Problem \cite{CS05}
are useful for the problem of packing weighted squares into
rectangular bins (the bins may have different dimensions) with large resources,
where {\em large resource} means that the height of a bin 
is much larger than the width,
and give a Polynomial Time Approximation Scheme (PTAS)
for the above problem,
then apply the PTAS to the problem of packing
squares with profits into a rectangular bin and
get a   $\frac65+\epsilon$ approximation algorithm.
For packing squaures into a rectangular bin,
we first introduce a simple algorithm by the techniques used in 
  2D bin packing \cite{BS06} problems and show that  
 if there are too many large squares packed or 
  the rest area for packing squares is not small,
 the algorithm has a nice performance,
then we focus on the case in which there are a few large squares packed
 and the rest area for small squares is also few.
 We propose a novel  approach of packing  a few large items such
 that the packing does not affect the future small items
 packing too much,  and call it {\em corner} packing.
 For packing small squares into the rectilinear polygons which is  generated
 after packing large squares into the bin,
 \begin{itemize}
 \item we first dissect the polygons into rectangular blocks such that
 the optimal value of packing small squares into the blocks
 is near the optimal value of packing small squares into the polygons,
 
 \item then call the method used in Multiple Knapsack Problem \cite{CS05},
       to guess one sublist which has a feasible packing and profit at least 
       $(1-\epsilon)OPT_b$, where $OPT_b$ is the optimal value for packing
       small squares into the blocks,

 \item lastly,  we exploit the techniques used in strip
       packing \cite{FGJ04, KR00} to pack items in each block.
 \end{itemize}

 {\bf Worst Case Ratio:}
 We adopt the standard measure {\em worst case ratio} to evaluate
 approximation algorithms. For any input list $L$, 
 let $A(L)$ be the total profit packed by approximation algorithm 
 $A$ and $OPT(L)$ be the  optimal value. 
 The {\em worst case ratio} of algorithm $A$ is thus defined as
 $$ R_{A}=\sup\limits_L
\frac {OPT(L)}{A(L)}.$$

 {\bf \em p($\cdot$), w($\cdot$)}:
 Given a square $q$, we use $p(q)$ and $w(q)$ to denote its profit and
 area respectively.
 And given a list of squares $L=(q_1,\dots,q_n)$, we define
 $p(L) = \sum_{i=1}^n p(q_i)$ and $w(L) = \sum_{i=1}^n w(q_i)$.

\section{Packing Squares into Rectangular Bins with Large Resources}
\label{sec:resources}
\para INSTANCE: Given an input list $S$ of $n$ squares with profits
           and a set of rectangular bins $B=(B_1,B_2,\dots,B_c)$ 
           where $B_i=(w_i,h_i)$ and 
           $ max\{w_i, h_i\} \ge \epsilon^{6^i-1}$ for all $i$,
           $c$, $\epsilon$ are constants. 

\para  OBJECTIVE: Maximize the total profit packed   in $B$.

Based on  the ideas from the seminal papers \cite{CS05,KR00},
we give a PTAS for the above problem.
There are three steps in the PTAS.
We first guess a subset of squares which can be packed into $B$
and whose total profit is near the optimal value through
the technique of rounding the input instance into $O(\epsilon^{-2} \ln n)$ classes.
Then for each bin, we guess  the number of items packed in that
bin from  each class such that
our guess is also near the optimal solution,
i.e.,we do not lose too much profit.
After matching items into bins, 
 we use the strip packing algorithm to pack items in each bin.

\para{\bf Rounding and guessing:}
Here, we consider square packing.
Since there is an natural order relation between any two squares,
the techniques used in Multiple Knapsack problem \cite{CS05}
are useful for square packing too.
We first round  the instance into a well structured instance
which has  $O(\epsilon^{-1} \ln n)$ distinct  profits,  and more items in each profit class  have at most 
       $(\epsilon^{-1})$ distinct sizes (side length).
Then we select a subset items which can be packed into the bins
and has the profit as least $(1-\epsilon)$ time the optimal solution.
But, if the items are rectangles, we do not have the above result, since
there is not an order relation between any two rectangles.
\begin{lemma}
 Given an above instance $I = (B,S)$ with $n$ items,
  in polynomial time $v = n ^{O(1/\epsilon^{3})}$,
 we can obtain instances $I_1,\dots, I_v$ such that
 \begin{itemize}
  \item  $I_j = (B, S_j)$ for $1 \le j \le v$, where $S_j$ is a
   sublist of list $S$.
 \item For $1 \le j \le v$, items in $S_j$ have $O(\epsilon^{-1} \ln n)$ 
       distinct  profits, 
      and more items in each profit class  have at most 
       $(\epsilon^{-1})$ distinct sizes (side length).
  \item There is an index $j$, $1 \le j \le v$,
        such that $S_j$  has a feasible packing in $B$ and 
        $p(S_j) \ge (1-O(\epsilon))OPT(I)$.
 \end{itemize}
 \label{lemma:guessing}
\end{lemma}

\begin{proof}
We show how to construct instance $I_1,\dots, I_v$ from  $I = (B,S)$
such that one of them satisfies the conditions.
There are four steps, which are basically from \cite{CS05}.
\begin{itemize}
\item Guess a value $\mathcal{O}$ such that 
$(1-\epsilon)OPT \le \mathcal{O} \le OPT$.
\item Round down the profits of items into $O(\epsilon^{-1} \ln n)$ classes
      such that 
      $\frac{p_j}{1+\epsilon} \le p_{j}^{-} \le p_j$,
      where $p_j$  is the $j$th item's profit and 
      $p_j^{-}$ is the one after rounding down, 
      where $n$ is the number of items.
\item Guess a set of sublists based the value $\mathcal{O}$ and the rounded
      profits of items such that
      one of them is feasible to pack into the set of bins (blocks) $B$
      and its total profit is at least $(1-O(\epsilon))OPT$.
\item Using the techniques in bin packing \cite{FL81},
      in each distinct profit class,
      reduce the number of distinct sizes into $O(\epsilon^{-1})$ 
      such that we lose the profit at most   $O(\epsilon)OPT$.
      Hence, totally, each of sublists has  $O(\epsilon^{-2} \ln n)$ 
      distinct size values and profits.     
\end{itemize}
 Next, we give the details for the above four steps.
 First, we show how to guess $\mathcal{O}$.
 Given a sufficiently small constant $\epsilon >0$,
 let $p_{max}$ denote the largest value among item profits.
 We know the optimal solution is bounded by $n \cdot p_{max}$.
 So we guess $\mathcal{O}$ from the set 
 \[
\{p_{max}(1+\epsilon)^i| 0 \le i \le 1+ \ln_{1+\epsilon} n\le 2\epsilon^{-1}\ln n\}.\]
 ($1+ \ln_{1+\epsilon} n\le 2\epsilon^{-1}\ln n$ follows from
  $\ln(1+\epsilon) \ge \epsilon -\epsilon^2/2 \ge \epsilon/2$.)
 Therefore, one of the values in the above set is 
 guaranteed to satisfy the desired property for $\mathcal{O}$.

 Given a value $\mathcal{O}$ such that
 $\max\{p_{max},(1-\epsilon)OPT\} \le \mathcal{O} \le OPT$,
 then we show how to massage the given instance into a more structured
 one has few distinct profits.
 \begin{enumerate}
 \item Discard all items with profits at most $\epsilon \mathcal{O}/n$.
 \item Consider the other items and divide all profits by  
       $\epsilon \mathcal{O}/n$ such that after scaling each profit
       is at most $n/\epsilon$.
 \item Round {\em down} the profits of item to the nearest power of
       $(1+\epsilon)$.
 \end{enumerate}
 It is easily seen that only an $O(\epsilon)$ fraction of the 
 optimal profit lost by the above transformation.
 Since $(1+\epsilon)^i \le n/\epsilon$,
 we have 
 \[
  i \le 2\epsilon^{-1} \ln n/\epsilon \le 4 \epsilon^{-1} \ln n.
\]
 The last inequality follows from $n/\epsilon \le n^2$.
 Therefore, we can transform the instance into a new instance with 
 $O(\epsilon^{-1} \ln n)$ distinct profits
 such that only an $O(\epsilon)$ fraction
 of the optimal profit is lost.

 Next we show how to guess the items to pack on the instance 
 with $O(\epsilon^{-1} \ln n)$ distinct profits.
 Let $h \le 4\epsilon^{-1} \ln n +1$ be the number of distinct
 profits in our new instance.
 We partition the input set of squares $\mathcal{S}$ into $h$
 set $S_1,...,S_h$ with items in each set having the same profit.
 Let $U$ be the items chosen in some optimal solution and  let
 $U_i = S_i \cap U$.
 Recall that we have an estimate $\mathcal{O}$ of the optimal value.
 If  $p(U_i) \le \epsilon \mathcal{O}/h$, we ignore the set $S_i$;
 no significant profit is lost.
 Hence we can assume that 
 $ \epsilon \mathcal{O}/h \le p(U_i) \le \mathcal{O} $
 and approximately guess the value $p(U_i)$ for $1 \le i \le h$,
 where $P(U_i)$ is the total profit in $U_i$.
 More precisely, for each $i$ we guess
a value $k_i \in [h/\epsilon^2]$ such that 
 \[
   k_i (\epsilon^2 \mathcal{O}/h) \le p(U_i) \le  
(k_i +1) (\epsilon^2 \mathcal{O}/h),
\]
 where $[h/\epsilon^2]$ stands for the set of integers
 0, 1,...., $\lfloor h/\epsilon^2\rfloor$.

 We show how the numbers $k_i$ enable us to identify the items
 to pack and then show how the values  $k_1,....,k_h$ can be 
 guessed in polynomial time.
 Given the value $k_i$ we order the items in $S_i$ in increasing
 size values (side length).
 Let $a_i$ denote the profit of an item in $S_i$.
 If $a_i \le \epsilon \mathcal{O}/h$,
 pick the largest number of item from this ordered set
 whose cumulative profit does not exceed $k_i(\epsilon^2 \mathcal{O}/h)$.
 If $a_i > \epsilon \mathcal{O}/h$ we pick the smallest number of items,
 again in increasing order of side lengths,
 whose cumulative profits exceeds $k_i(\epsilon^2 \mathcal{O}/h)$.
 The choice of items is thus completely determined by the choice of the $k_i$.
 For a tuple of values $k_1,....,k_h$,
 let $U(k_1,...,k_h)$  denote the set of items packed as described above.
 
 From the above selection,
 there exists a valid tuple $(k_1,...,k_h)$  with 
 each  $k_i \in [h/\epsilon^2]$ such that 
 $U(k_1,...,k_h)$ has a feasible packing in $B$ and 
 $p(U(k_1,...,k_h)) \ge (1-\epsilon)  \mathcal{O}$.

 Now we show that the values $k_1,....,k_h$ can be guessed
 in polynomial time. Before that, we introduce a useful claim.
 \begin{claim} \cite{CS05}
 Let $f$ be the number of $g$-tuples of non-negative integers
 such that the sum of tuple coordinates is equal to $d$.
 Then $f=\binom{d+g-1}{g-1}$.
 If $d+g \le \alpha g$ then $f=O(e^{\alpha g})$.
 \label{claim:counting}
 \end{claim}
 By Claim \ref{claim:counting},
 the number of $h$-tuples $(k_1,...,k_h)$ with 
 $k_i \in [h/\epsilon^2]$ and $\sum_i k_i \le h/\epsilon^2$
 is $O(n^{O(\epsilon^{-3})})$.

 Next we show how to reduce the number of distinct sizes (side length)
in each profit class. The basic idea is the one used in 
 approximation schemes for bin packing \cite{FL81}.
 Let $A$ be a set of $g$ items with identical profit.
 We order items in $A$ in non-decreasing sizes and 
 divide them into $t = (1+1/\epsilon)$ groups $A_1, ..., A_t$ with
  $A_1, ..., A_{t-1}$ containing $\lfloor g/t\rfloor$ items each
  and $A_t$ containing $(g \mod t)$ items.
 We discard the items in $A_{t-1}$ and for $i < t-1$ 
 we increase the size of every item in $A_i$ to the size of 
 the smallest item in $A_{i+1}$.
 Since $A$ is ordered by size,
 no item in $A_i$ is larger than the smallest item in $A_{i+1}$
 for each $1 \le i < t$.
 It is easy to see that if $A$ has a feasible packing
 then the modified instance also has
 a feasible packing.
 We discard at most an $\epsilon$ fraction of the profit and 
 the modified sizes have at most $2/\epsilon$ distinct values.
 Applying this to each profit class
 we obtain an instance with $O(\epsilon^{-2} \ln n)$ distinct size
 values.

 Hence, we have this lemma.
\end{proof}

\para{\bf Distributing the selected items into each bin}
 
After guessing a polynomial number of sublists,
next we consider how to distribute the selected items in each sublist
into bins. Easily to see, 
the possibilities to assign the selected items into bins
is bounded by $c^{n}$,
which is an exponential size of $n$,
where $c$ is the number of bins
and $n$ is the number of items to be packed.
But we can guess a subset from the selected items in a polynomial time
such that the total profit in the subset is near the optimal solution.

After step 1, we have $(\epsilon^{-2} \ln n)$ classes in the input instance.
 Let $k_i$ be the number of items of the $i$th  class
 and let $l_{i}^{j}$ be the number of items of the $i$th  class
 packed in the $j$th bins, where $1 \le j \le m$.
\begin{lemma}
\label{lemma:mapping}
 We can guess a set of numbers $h_i^{j}$ in polynomial time
 such that $(1-\epsilon)l_{i}^{j} \le h_i^{j} \le l_i^{j}$,
 where $1\le \epsilon^{-2} \ln n$ and $1 \le j \le c$ and 
 $c$ is the number of bins.
\end{lemma}
\begin{proof}
 For the $j$th bin,  we guess $h_i^{j}$ items from the $i$th class.
 If $k_i \le \frac{c}{\epsilon(1+\epsilon)}$ then 
 we can guess a number $h_i^{j}$ such that $h_i^j = l_i^{j}$ in 
 $O(\frac{c}{\epsilon(1+\epsilon)})$ time.
 Else, we guess a number $h_i^{j}$ from
 the set $\{ \lfloor (1+\epsilon)^{x} \cdot \frac{\epsilon k_i}{c}\rfloor |
 x= 1,2,\dots\}$
 such that $(1-\epsilon)l_{i}^{j} \le h_i^{j} \le l_i^{j}$.
 Since $h_i^{j} \le k_i$, 
 the number of guesses required to obtain a single $h_i^{j}$ is 
 bounded by $g= \log_{1+\epsilon}c/\epsilon \le O(\epsilon^{-2} \ln c)$,
 for each class, 
 the total number of guesses for all $h_i^{j}$ is bounded 
 by $g^{c} \le O(\epsilon^{-2c}c^{c})$, where $1 \le j \le c$.
 Therefore 
 for all   the $O(\epsilon^{-2}\ln n)$ size classes
 the total number of guesses for is 
 bounded by $n^{\epsilon^{-2}}$,
 which is a polynomial of $n$, where $c$ and $\epsilon$ are  constants.
\end{proof}
 Since all the items  in each size class  have
 the same profit 
 and by Lemma \ref{lemma:mapping} we have $h_i^j \ge (1-\epsilon)l_i^{j}$,
 there exists one  assignment
 which is feasible to $B$ and keeps at least $(1-\epsilon)$ times the profits.
 Next we consider how to packing items into each bin.

\para{\bf Packing each sublist into each bin}:
In each bin, we have the following property 
\[
\frac{\max\{w,h\}}{\min\{w,h\}} \ge \epsilon ^{-4},
\]
then  the techniques used in
\cite{KR00, FGJ04} are helpful to pack all squares into the bins.

First, we first give an important lemma for packing
squares into a bin with large resource, called {\em cutting technique}.
  \begin{lemma}
  Given an input list $L$ of squares with sides at most $\epsilon$
  and two rectangular bins $B_1=( 1, a)$, $B_2 =(1 + 2\epsilon, a)$,
  then  
  \[
  (1- 4\epsilon ) \cdot OPT(L, B_2) \le OPT(L,B_1),
  \]
  where $OPT(L,B)$ is the optimal value for packing list $L$ into  bin $B$.
  \label{lemma:cutting}
 \end{lemma}
 \begin{proof}
  Now we construct a packing in bin $B_1$ from an optimal packing
  in bin $B_2$ and prove its profit is at least
  $(1- 4\epsilon ) \cdot OPT(L, B_2)$.
  
  Consider an optimal packing in bin $B_2$,
  we cut  $B_2$ into $\lfloor \frac{1}{4\epsilon} \rfloor$ pieces 
  of slices, say $S_1, S_2, \dots, S_{\lfloor \frac{1}{4\epsilon} \rfloor}$
  respectively, such that every slice has an exact width  $ 4 \epsilon$ 
 (except the last one),
  shown as fig.~\ref{fig:strips}.
  (Note that some squares may be cut into two parts,
   one part in $S_i$ and another part in $S_{i+1}$).
 \begin{figure}[htbp]
  \begin{center}
  \includegraphics[scale=0.5]{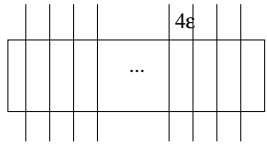}
  \caption{Cutting  bin $B_2$  into  slices}
  \label{fig:strips}
  \end{center}
\end{figure}
  Then   we  find a slice $S_i$ such that $p(S_i) \le 4\epsilon OPT(L,B_2)$
  and remove all squares {\em completely} contained  in slice $S_i$
  if any.
  Observe that after the above removal,
  all  squares remaining in bin $B_2$ can be packed into $B_1$.
  Hence,  $OPT(L,B_1) \ge (1- 4\epsilon ) \cdot OPT(L, B_2)$.
 \end{proof} 

\begin{lemma} \cite{KR00}
 There is an algorithm $A$ which, 
 given a list $L$ of $n$ square and  a positive $\epsilon$,
 produces a packing of $L$ in a strip of width 1 and height $A(L)$ 
 such that $A(L) \le (1+\epsilon)Opt(L) + O(1/\epsilon^2)$.
 \label{lemma:stripPTAS}
\end{lemma}
\begin{lemma}
For packing small squares into a constant number of bins, 
for each bin, if $\frac{\max\{w,h\}}{\min\{w,h\}} \ge \epsilon ^{-4}$,
then  there is a polynomial time algorithm
with an output at least  $ (1-O(\epsilon))OPT$,
where $OPT$ is the optimal value.
 \label{lemma:manyresource}
\end{lemma}
\begin{proof}
 Given an instance $I = (B,S)$,
 where $S$ is the set of  small squares with profits
 and $B=(B_1,B_2,\dots,B_c)$ is the set of rectangular bins,
 by the method in Lemma \ref{lemma:guessing},
 we  guess a subset $S_j \subseteq S$ such that
 $S_j$  has a feasible packing in $B$ and  $p(S_j) \ge (1-O(\epsilon))OPT(I)$.

In an instance $I_p = (B,S_p)$,
we first guess $h_i^j$ by Lemma \ref{lemma:mapping}.
Then according to $h_i^{j}$ value
we assign the items to each bin and 
 use the APTAS in Lemma \ref{lemma:stripPTAS}
to pack items in each bin, where $1 \le j \le c$.
If in each bin $(w,h)$ 
the height used by the APTAS in Lemma \ref{lemma:stripPTAS}
is bounded by $(1+\epsilon)\max\{w,h\} + O(\min\{w,h\}/\epsilon^{2})$,
then we keep the assignment otherwise reject the assignment.
Since there is a $S_p$ such that $S_p$ has a feasible packing in $B$.
After all the guesses,
there is at least one assignment remained.
For the assignment, we apply the  APTAS in Lemma \ref{lemma:stripPTAS}
and the cutting techniques in Lemma \ref{lemma:cutting} 
 such that  in each bin, 
 the profit keeps at least $(1-\epsilon)$ times the optimal value.

 Hence we have this lemma.
\end{proof}
 
\section{Previous Algorithms for Packing Squares}
 Based on previous techniques used for 2D packing problem 
 \cite{BS06,  Harren06} and
 the {\em greedy} packing (which is given in appendix), 
  we  introduce a simple algorithm $A_1$ which is implied in 
 \cite{BS06,  Harren06}
 for packing a set of squares into a bin $(1,h)$,
 where $h \ge 1$.
There are two steps in $A_1$:
  first group all squares by their sizes and guess one group
   which does not significantly  affect the optimal packing  and delete it
   from the input list,
  then  pack large items by enumeration, lastly
   append small items in the ``gap'' of the bin.
 Next, we give the details of the two steps. 

\noindent {\bf Grouping}:   
For an integer $k = \lceil \frac{1} {\epsilon}\rceil$, 
where $ \epsilon < (2h+2h^2)^{-1}$ is  sufficiently small
and $h \ge 1$ is the bin height,  
we select $k$ points in the region (0,1], 
$P_1,\dots,P_k$ as follows 
\[
P_i =  \epsilon^{6^i} \textrm{ and }  1 \le i \le k. 
\]
Then the interval (0,1] is divided into $k+1$
intervals, $I_1,\dots,I_{k+1}$, 
where $I_i=(P_i,P_{i-1}]$, $2 \le i \le k$, $I_1 = (P_1,1]$ 
and $I_{k+1} = (0,P_k]$.

\noindent{\bf Notation:}
 In the following, given a  list $L$ of squares, $L_i$ denotes
the list in which all square's sides are in interval $I_i$, 
$w(L_i)$ denotes the total area of $L_i$,
$p(L_i)$  the total profits of  $L_i$,
$|L_i|$  the number of squares in  $L_i$,
 where $1 \le i \le k+1$.
 
\noindent{ \bf Packing:}
 \begin{enumerate}
 \item  Guess an index $i$ such that 
        $OPT(L-L_i) \ge (1-\epsilon)OPT(L)$.
  \item Get all   feasible packing for $L_{i-1} \cup \dots \cup L_1$, 
        pack each of them 
             into the bin, then
          partition the uncovered space into  rectangular bins(blocks) 
        in the method \cite{BS06} and append $L_{i+1} \cup \dots \cup L_{k+1}$
          into these bins by the {\em Greedy} algorithm.
 \item Output the one with the largest profit.
 \end{enumerate}

    Since there are $k+1$ sublists $L_1, \dots, L_{k+1}$ in $L$,
    then the guess in step 1 of $A_1$ is always feasible,
     where $k =  \lceil \frac{1}{\epsilon} \rceil$.
   After selecting an index $i$,
  we define all items in $L_{i-1} \cup \dots \cup L_1$ as 
  {\em large} items and all items in $L_{i+1} \cup \dots \cup L_{k+1}$
  as {\em small} items.
  Note that if $i=1$ then there are no {\em large} items,
  and if $i= k+1$ then there are no {\em small} items.

  $A_1$'s worst case ratio is related to 
  the number of {\em large} items in $L_{opt}$ and 
  the rest area for small items,
  where  $L_{opt}$ is a sublist of $L$ to produce an optimal solution.

 \begin{fact} \cite{BS06}
  Given   {\em large} items with sides larger than
   $ \epsilon^{6^{i-1}}$ which can be packed in the bin  $(1,h)$,
  and {\em small} items with sides at most $ \epsilon^{6^i}$,
 where $i \ge 1$,
 if the total area of  all the squares   is at most 
 $h -  \epsilon^{4 \times 6^{i-1}-1} $, where $2h(1+h) < \epsilon^{-1}$,
 then  all  can be packed in the bin.
  \label{fact:append}
 \end{fact}

\begin{lemma}\cite{BS06,Harren06}
  After packing large items,
 if the rest area in the bin is at least $\epsilon^{-1}\delta$,
 then $A_1(L) \ge (1-2\epsilon)OPT(L)$, 
where $\delta=\epsilon^{4\times 6^{i-1}-1}$.
 \label{lemma:restarea}
\end{lemma}

 
\begin{lemma} \cite{Harren06}
 Let $m $ be the number of {\em large} items in $L_{opt}$.
 i)  $A_1(L) \ge (1- \epsilon) OPT(L)$ if $m= 0$;
 ii) else 
 $ A_1(L) \ge \frac{m}{m+1}(1-\epsilon)OPT(L)$.
 \label{lemma:m+1} 
\end{lemma}

\begin{lemma}  \cite{Harren06}
 Algorithm $A_1$ is ran in polynomial time of $n$.
 \label{lemma:time}
\end{lemma} 

\section{  Corner packing}
\label{Sec:NFDHCorner}
 To pack squares into a rectangular bin,
 there are a lot of approaches, the most two studies are 
 {\em NFDH} \cite{BS06} and {\em BL}.
 In this section, we first give a new approach, called {\em Corner} packing,
 which includes the above two approaches.
 Then we analyze the {\em corner} packing later and show that
   it is  one of the key points for improving
 the worse case ratio.

 During packing  squares into the rectangular bin,
 the uncovered space of the bin may get into the rectilinear polygons.
The {\em corner} packing (shown as in Fig.~\ref{fig:NFDHDeco}(b))
  can be regarded as a sequence of packing.
 Every time when one square is packed into the bin, we  obey
 the following rules:
 \begin{itemize}
\item select one vertex of the current rectilinear polygons 
   at which the interior angle is 90 degrees,
\item place the square such that one of its corners coincides with
    the vertex we selected. After packing, we get the new rectilinear
    polygons.
 \end{itemize}
Note that both NFDH and  BL \cite{jz04}  belong to {\em Corner} packing,
where BL packing is to pack squares in the bin as bottom as possible
then as left as possible.
 
\begin{lemma} 
 \label{lemma:vertex}
 Assume $n$ squares are packed in the bin by corner packing,
 then \\
 i) there are at most $4+2n$ vertices of all the rectilinear polygons, \\
 ii) there are at most $2^n (n+1)!$ possibilities to pack these $n$
     squares in the bin by corner packing.
\end{lemma} 
(refer to the proof in the appendix).


 \section{ A Refined Algorithm for Packing  into a Rectangular Bin}
 Let  $m $ be the number of {\em large} items in $L_{opt}$,
 where $L_{opt}$ is a sublist of $L$ to get the optimal solution.
 By Lemma \ref{lemma:m+1},
 if $m$ is very large, 
 then algorithm $A_1$ has a good performance.
 So,  to improve  algorithm $A_1$,
 we have to study the case in which $m$ is very small.
 Note that when the bin is a unit square,
 the situation becomes a relatively simple.
 Since when $m=1$ we can transform the original packing into a special
 strip packing;
 when $m=2,3$ we can estimate there must be much more space left for
 small squares than the wasted area.
 This is  the main idea in Harren's paper\cite{Harren06}.
 If the bin is not longer a unit square, 
 his algorithm does not work very well.
 To improve algorithm $A_1$,
 we are faced with two problems:
 \begin{itemize}
 \item  How to pack a few large items such that the packing does not
        affect too much the future small items packing?
        (how to allocate large items in the bin.)
 \item  How to pack small items in the gaps (rectilinear polygons)
        generated after packing large items?
 \end{itemize}
 Next, we  give our solutions
 for the above questions and propose a refined algorithm called $A_2$ with
 the worst case ratio $\frac65 + \epsilon$.

\subsection{ Packing a few large items}
Recall that, given an $i >1$,
if a square's  side length is at least $\epsilon^{6^{i-1}}$ then
 it is called {\em large},
else its side length is at most $\epsilon^{6^{i}}$ then
it is called {\em small}.
And there is a gap between large items and small items,
which is very important for packing large items.

Next we show that {\em corner} packing is a good packing which
does not  significantly affect future small square packing
when there are a few large items.


\begin{lemma}
Let $m $ be the number of {\em large} items in $L_{opt}$.
If $m \le 4$, then $(1-\epsilon)OPT(L) \le OPT(L, *)$,
where $OPL(L)$ is the optimal value of packing $L$ into the bin
and $OPT(L,*)$ is  the optimal value of packing $L$ into the bin such
that all large squares are packed by {\em corner packing}. 
\label{lemma:corner}
\end{lemma}

 \begin{proof}
  Here, we just give the proof when $m=4$, since the proof for $m=1,2,3$ is 
  involved.
  Let $a,b,c,d$ be the four large squares in an optimal packing $L_{opt}$.
  Without loss  of generality assume $a,b,c,d$ are placed in  the bin 
  as Fig.~\ref{fig:optimal1}(1). 
  Note that a large item has side at least $\epsilon^{6^{i-1}}$
  and a  small item has side at most  $\epsilon^{6^{i}}$,
  where $i \ge 1$.
  We cut the bin into three parts, 
  two rectangular blocks $I=(x_1,y_1)$, $II=(x_2,y_2)$ and 
  a rectilinear polygon $P$ as shown as Fig.~\ref{fig:optimal1}(2).
    \begin{figure}[htbp]
  \begin{center}
  \includegraphics[scale=0.6]{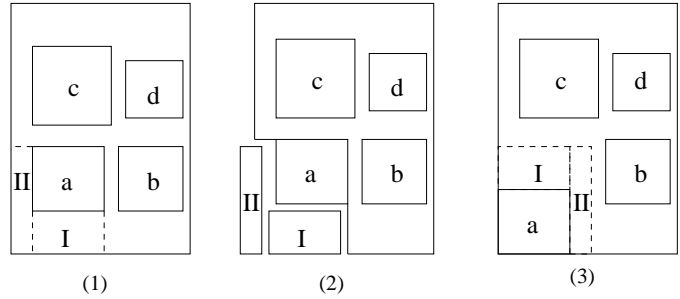}
  \caption{ An optimal packing vs. its corner packing}
  \label{fig:optimal1}
  \end{center}
\end{figure}
 Now we define two new rectangular blocks 
 $I_{\epsilon}=(x_1+ 2\epsilon^{6^i},y_1)$
 and $II_{\epsilon}=(x_2, y_2+2\epsilon^{6^i} )$.
 Then we have 
 \begin{equation}
   OPT(L,B) \le OPT(L, I_{\epsilon} \cup II_{\epsilon} \cup P),
  \label{eqn:expending}
 \end{equation}
  where $OPT(L,B)$ is the optimal value of packing $L$ into the bin $B$
  and  $OPT(L, I_{\epsilon} \cup II_{\epsilon} \cup P)$
  is the optimal value of packing $L$ into three rectilinear polygons
  $I_{\epsilon} \cup II_{\epsilon} \cup P$.
  This can be seen as follows,
  all squares packed into the bin $B$  
 as shown in Fig.~\ref{fig:optimal1}(1) can 
  be packed into three rectilinear polygons
  $I_{\epsilon} \cup II_{\epsilon} \cup P$.
  By Lemma \ref{lemma:cutting},
  for any list $L$ of small squares,
  we have 
 \[
   (1-4\epsilon^{6^i})OPT(L, I_{\epsilon}) \le OPT(L, I) \textrm{ and }
   (1-4\epsilon^{6^i})OPT(L, II_{\epsilon}) \le OPT(L, II).
 \] 
 Then
 \begin{equation}
 (1-4\epsilon^{6^i}) OPT(L,I_{\epsilon} \cup II_{\epsilon} \cup P ) \le 
  OPT(L, I \cup II \cup P ).
  \label{eqn:upperbound}
\end{equation}
So,
 by (\ref{eqn:expending}), (\ref{eqn:upperbound}),
\[
  (1-4\epsilon^{6^i}) OPT(L,B) \le OPT(L, I \cup II \cup P ). 
\]
And we have $OPT(L, I \cup II \cup P ) \le OPT(L,*a)$,
where $OPT(L,*a)$ is the optimal value of packing $L$ into the bin 
with $a$ at one corner, shown as in Fig.~\ref{fig:optimal1}(3).
Hence
\begin{equation}
 (1-4\epsilon^{6^i}) OPT(L,B) \le  OPT(L, *a ). 
 \label{eqn:a-corner}
\end{equation}
By the similar proof, we have 
\[
  (1-4\epsilon^{6^i})OPT(L,*a) \le OPT(L, *ab )
\]
and 
 \[
  (1-4\epsilon^{6^i})OPT(L,*ab) \le OPT(L, *abc ), \textrm{ }
  (1-4\epsilon^{6^i})OPT(L,*abc) \le OPT(L, *abcd )
\]
 where $OPT(L, *X)$ is the optimal value of packing a set $X$ into the bin 
 with all the items in $X$ at corners of the bin.
 Therefore, we have 
 \[
 OPT(L, *abcd) \ge (1-4\epsilon^{6^i})^4   OPT(L,B) \ge  (1-\epsilon) OPT(L,B).
\]
 The last inequality follows from $\epsilon \le 1/2 $.

Hence, this lemma holds.
 \end{proof}

\subsection{ Packing small items into rectilinear polygons}
 After packing large items,
 the uncovered space in the bin  may be a set of rectilinear 
 polygons.
 Our strategy for packing small squares into the polygons are below:
\begin{itemize}
 \item Dissect the polygons into rectangular blocks such that
       the optimal value of packing small squares into the blocks
       is at least $(1-\epsilon)OPT_p$, where 
        $OPT_p$ is the optimal value for packing
       small squares into the polygons. 
 \item To pack small items into blocks,
       we use the PTAS in Section \ref{sec:resources} of
       packing squares into rectangular bins with large resources. 
       
\end{itemize}

\noindent{\bf Dissection:}
 After packing few large squares into the bin by {\em corner} packing,
 we dissect the rectilinear polygons into rectangular blocks,
 such that the dissection does not affect the optimal packing insignificantly.

\begin{lemma}
 If there are at most 4 large squares packed,
 and  the total area of the large squares packed
  is at least $h -\epsilon^{4\times 6^{i-1}-2}$, 
 then there exist a dissection of the polygons (the uncovered space 
 of the bin) into blocks such that 
\[
   OPT_b \ge (1-\epsilon)OPT_p, 
\]
where $OPT_b$ ($OPT_p$) is the optimal value of packing small squares
into blocks (polygons).
\label{lemma:dissect}
\end{lemma}
\begin{proof}
 In this proof, we just give our dissection for four large squares packed,
 shown as in Fig.~\ref{fig:our4} and \ref{fig:2cases} (by dotted lines),
 since the number of large squares is less than 4,
 we have the similar dissection.

   \begin{figure}[htbp]
  \begin{center}
  \includegraphics[scale=0.5]{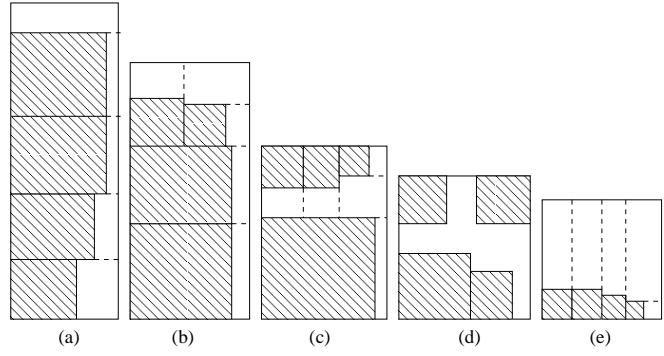}
  \caption{Possible packing}
  \label{fig:our4}
  \end{center}
\end{figure}
 \begin{figure}[htbp]
  \begin{center}
  \includegraphics[scale=0.5]{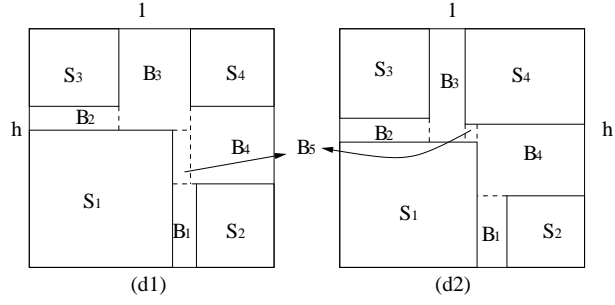}
  \caption{Two possibilities for the case (d)}
  \label{fig:2cases}
  \end{center}
\end{figure}

Except for case (d),
if  the rectilinear polygons  are dissected into blocks
 as shown in Fig.~\ref{fig:our4},
then we have a set of five rectangular blocks 
$B=\{B_i\}$, 
where $B_i=(w_i,h_i)$ and $w_i \le h_i$ for $1 \le i \le 5$.
(Otherwise we can exchange $w_i$ and $h_i$)
 Using the same techniques in Lemma \ref{lemma:corner},
 we define a new set of five
 blocks $B^{'}=\{B_i^{'}\}$, 
where $B_i^{'}=(w_i ,h_i + 2 \epsilon^{6^{i}} )$ for $1 \le i \le 5$.
Let $P$ be the polygon(s) after packing large squares in the bin,
and $L$ be a list of small squares.
Since each small square has side length at most  $\epsilon^{6^i}$,
 we have 
 \[
  OPT(L,B') \ge OPT(L,P) \ge OPT(L,B),
 \]
 where OPT(L,X) is the optimal value for packing $L$ into $X$.
Since  each large square has side length at least  $\epsilon^{6^{i-1}}$ 
and $h_i \ge w_i$,
we have $h_i \ge \epsilon^{6^{i-1}}$.
 By Lemma \ref{lemma:cutting}, we have
\[
   OPT(L,B)  \ge (1-\epsilon)OPT(L,B^{'}) .
\]
Hence we have  
$OPT_b = OPT(L,B) \ge (1-\epsilon)OPT(L,P) =(1-\epsilon)OPT_p$.

Next, we study  the case (d) of Fig.~\ref{fig:our4}
and prove that our strategy  shown in Fig.~\ref{fig:2cases}
still works.
There are two possibilities for the case (d).
We assign $S_1, S_2, S_3, S_4$ to the four large squares
as shown in  Fig.~\ref{fig:2cases},
where $S_i=(s_i,s_i)$.
And the polygon is dissected into 5  blocks $B_1,\dots,B_5$,
where $B_i=(w_i, h_i)$.

From our dissections in  Fig.~\ref{fig:2cases}, (by dotted lines),
we have 
\[
   \max\{w_i,h_i\} \ge \epsilon^{ 6^{i-1}}, \textrm{ for  } 1 \le i \le 4.
\]
And since the total area of the large items in the bin is at least 
  $h -\epsilon^{4\times 6^{i-1}-2}$, i.e.,
 the total area of the blocks is at most  $\epsilon^{4\times 6^{i-1}-2}$.
We have 
\[
  \min\{w_i,h_i\} \le \frac{\epsilon^{4\times 6^{i-1}-2}}{\epsilon^{ 6^{i-1}}}
  \le \epsilon^{ 2 \times 6^{i-1}}.
\]
The last inequality follows from $i \ge 2$.

Let $\delta = \epsilon^{ 6^{i-1}}$,
next we prove that $w_5 \le \delta^{ 2 }$ in the cases (d1) and (d2).
It is trivial to see $w_5 \le \delta^{ 2 }$ in the case (d1), 
since $w_5 \le w_1 =  \min\{w_1,h_1\}$.
Now, we consider the case (d2),
since $s_1 + s_2 \le 1$ and $s_1 -s_2 \le h_4 \le \delta^{2}$,
we have 
\[
  s_1 \le \frac{1+\delta^2}{2}.
\]
Since $s_3 + s_4 \le 1$ and $s_4 -s_3 \le h_2 \le \delta^{2}$,
we have 
\[
  s_4 \le \frac{1+\delta^2}{2}.
\]
So,
\[
  w_5 = s_1 + s_4 -1 \le 1+ \delta^2 -1 = \delta^2. 
\] 

Therefore, we have  in the cases (d1) and (d2),
 $w_5 \le \delta^2 \le \epsilon w_4$ and $h_5 \le h_4$,
i.e., to compare with block $B_4$,
block $B_5$ is sufficiently small and can be ignored.
So
\[
  OPT(S, B^{-}) \ge (1-O(\epsilon))OPT(S,B),
\]
where $S$ is the set of small squares, 
      $B^{-} = \cup_{i=1}^{4} B_i$,  $B = \cup_{i=1}^{5} B_i$.
Hence when we pack small squares into blocks $ \cup_{i=1}^{5} B_i$,
we just consider $ \cup_{i=1}^{4} B_i$.
Then by  the similar proof for other cases,
we  have  $OPT_b \ge (1-O(\epsilon))OPT_p$.


Hence, this lemma holds.
\end{proof}


 \subsection{ Algorithm $A_2$ and its analysis}
 Next, we give the details of  algorithm $A_2$
 which is similar to $A_1$.

 \para {\bf Description of  Algorithm $A_2$}
\begin{enumerate}
 \item Group items and guess an index $i$ such that 
        $OPT(L-L_i) \ge (1-\epsilon)OPT(L)$
       and  divide the remaining
       into two classes, say {\em large} and {\em small}, 
\item   Get all feasible packing of $L_{i-1} \cup \dots \cup L_1$,
            for each of them,
          \begin{enumerate} 
         \item  
                 if there are at least  4   items or 
                 the total area of items is at most 
                 $h -\epsilon^{4\times 6^{i-1}-2}$,
                 then  pack large and small squares
                 by algorithm $A_1$.
         \item else  locate large items as Fig.~\ref{fig:our4} 
               and divide the gaps into blocks as 
               Fig.~\ref{fig:our4} and \ref{fig:2cases},
               lastly  apply the method in Lemma \ref{lemma:manyresource} 
               for {\em small} items.
       \end{enumerate}
 \item Output the one with the largest profit.
 \end{enumerate}

\begin{theorem}
 For any input list $L$,
 $\frac{OPT(L)}{A_2(L)} \le   \frac65 (1+O(\epsilon))$,
 where $\epsilon$ is sufficiently small.
 \label{theorem:m<=4}
\end{theorem} 
\begin{proof}
 To consider an optimal packing solution $L_{opt}$,
 if there are at least 5 large items in $L_{opt}$ or
 the total area of large items in $L_{opt}$ 
 at most $h -\epsilon^{4\times 6^{i-1}-2}$,
 by Fact \ref{fact:append}, Lemmas 
\ref{lemma:restarea} and \ref{lemma:m+1}, 
 \[
   A_2(L) \ge \frac56(1-2\epsilon)OPT(L).
 \]
Else, the total area of the large items in the bin is at least 
  $h -\epsilon^{4\times 6^{i-1}-2}$
 and there are at most 4 packed.

 By the dissection of the polygons into rectangular blocks,
 shown as in Fig.~~\ref{fig:our4} and \ref{fig:2cases},
 in each block $(w,h)$, 
 we make sure that $\max\{w,h\} \ge \epsilon^{ 6^{i-1}}  $.
 So,  
\[
  \frac{\max\{w,h\}}{\min\{w,h\}} 
 \ge
 \frac{\max\{w,h\}}{\frac{\epsilon^{4\times 6^{i-1}-2}}{\epsilon^{ 6^{i-1}}}}
 \ge \frac{\epsilon^{2\times 6^{i-1}}}{\epsilon^{4\times 6^{i-1}-2}}
 \ge \epsilon^{-4}.
\]
The last inequality follows from $i\ge 2$.
(remember when $i=1$, there is no large item.)
By Lemmas  \ref{lemma:corner}, \ref{lemma:manyresource},
we  have 
\[
     A_2(L) \ge (1-O(\epsilon))OPT(L)
\]

By Lemma \ref{lemma:time}, the time complexity of 
Algorithm $A_2$  is a polynomial time of $n$.
Hence, this theorem holds. 
\end{proof}

 \section{Concluding remarks}\label{se:con}
 Note that  algorithm $A_2$ 
 can be expended to multi-dimensional cube packing.
 

\normalsize
\newpage
\section{ Appendix}
\subsection{The proof for Lemma \ref{lemma:vertex}}
 \begin{proof}
  We use induction to prove part i).
  When $n=0$, there are 4 vertices in the bin.
  When $n=1$, there are at most 6 vertices in the rectilinear polygon.
  So, we assume that when $n=k$, part i) holds, i.e.,
  after packing $k$ items in the bin,
  there are at most $4+2k$ vertices in the rectilinear polygons.
  When $n= k+1$, we use one of $4+2k$ vertices and generate at most 3 vertices,
  hence the total number of vertices is at most 
  \[
    4 + 2k -1 + 3 = 4 + 2(k+1).
  \]
  Then we can see  there are at most 
  $4 + 2(i-1)$ ways to pack the $i$-th square, where $i \ge 1$.
  Hence to pack $n$ items in the bin,
  there are at most 
  \[
   \prod_{i=1}^{n} (4 + 2(i-1)) = 2^n \prod_{i=1}^n (i+1) = 2^n (n+1)!
  \] 
  possibilities. 
 \end{proof}

 \subsection{ NFDH packing }
 {\em NFDH} (Next Fit Decreasing Height) \cite{MM68}.
 {\em NFDH} packing behaves as follows:
 First sort all squares by their  heights, 
 then pack them in the bin from the largest one level by level as shown 
 in Fig. ~\ref{fig:NFDHDeco}(a).
 In each level, pack them by Next Fit, namely,
 if the current level cannot accommodate the next item,
 then open a new  with height equal 
 to the current item's height.
 We repeat this procedure,
 until there is no space for a new level in the bin.
 Here is a key property of {\em NFDH}.
 \begin{figure}[htbp]
  \begin{center}
  \includegraphics[scale=0.7]{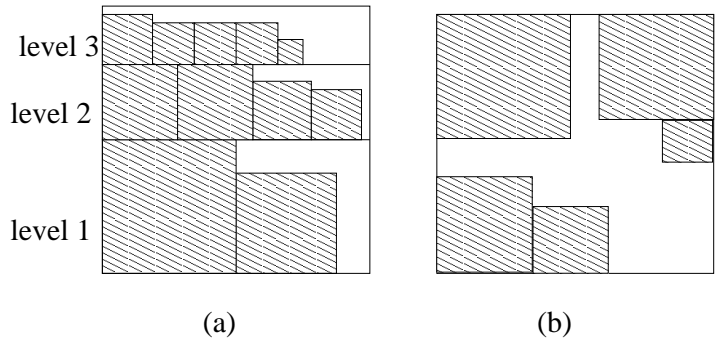}
  \caption{ NFDH and Corner packing}
  \label{fig:NFDHDeco}
  \end{center}
\end{figure}



 \para {\bf Greedy Algorithm}
\begin{enumerate}
 \item Sort the input list $L$ such that 
  $\frac{p(A_1)}{w(A_1)} \ge \cdots \ge \frac{p(A_k)}{w(A_k)}$.
 \item For $i$ from 1 to n do \\
        if $(a_i \ge \epsilon)$ and $( b_i \ge \epsilon)$ then
       \begin{enumerate}
       \item   Find a maximal index $m$ such that $(A_1,A_2,...,A_m)$
             can be packed into the current bin by {\em NFDH}
             and pack $(A_1,A_2,...,A_m)$ into $(a_i, b_i)$.
       \item   Then update list $L$ and re-index $L$.
               If  $L$ becomes empty then  finish packing.      
       \end{enumerate} 
 \end{enumerate}

\end{document}